%% file: itasec2020_main.tex
\RequirePackage[dvipsnames]{xcolor} 

\documentclass{easychair}

\usepackage[T1]{fontenc}
\usepackage[utf8]{inputenc}

\usepackage[english]{babel}

\usepackage{mathtools,amssymb,amsfonts,thmtools}
\usepackage{centernot}
\usepackage{booktabs}
\usepackage[all]{nowidow}
\usepackage{mathpartir}
\usepackage{stmaryrd}
\usepackage{listings}
\usepackage[draft,nomargin,inline,index]{fixme}
\fxusetheme{color}
\FXRegisterAuthor{l}{el}{\color{green} {\underline{lillo}}}
\FXRegisterAuthor{p}{ep}{\color{red} {\underline{pd}}}
\FXRegisterAuthor{m}{em}{\color{blue} {\underline{matteo}}}

\hypersetup{breaklinks=true}
\usepackage{url}
\usepackage{breakurl}
\usepackage{common}
\usepackage{seccomp}

\allowdisplaybreaks

\title{\mbox{Control-flow Flattening Preserves the Constant-Time Policy}%
\thanks{The first two authors have been partially supported by project PRA\_2018\_66
\emph{DECL\-ware:\ Declarative methodologies for designing and deploying applications} of the  Universit\`a di Pisa;
the third author by MIUR project PRIN 2017FTXR7S
\emph{IT MATTERS} (Methods and Tools for Trustworthy Smart Systems).}\\[-8pt]
{\small Extended Version}\vspace{-3mm}}
\titlerunning{Control-flow Flattening Preserves the Constant-Time Policy}

\author{ Matteo Busi\inst{1} \and Pierpaolo Degano\inst{1} \and Letterio
  Galletta\inst{2} \vspace{-1mm}}
\authorrunning{M.~Busi, P.~Degano and L.~Galletta}

\institute{
  Università di Pisa, Pisa, Italy ---
  \email{\{matteo.busi, degano\}@di.unipi.it}
\and
   IMT School for Advanced Studies, Lucca, Italy ---
   \email{letterio.galletta@imtlucca.it}
}

\begin{document}
\maketitle
\begin{abstract}
  Obfuscating compilers protect a software by obscuring its meaning and impeding
  the reconstruction of its original source code.
  The typical concern when defining such compilers is their robustness against
  reverse engineering and the performance of the produced code.
  Little work has been done in studying whether the security properties of a
  program are preserved under obfuscation.
  In this paper we start addressing this problem: we consider control-flow
    flattening, a popular obfuscation technique used in industrial compilers,
  and a specific security policy, namely constant-time.
  We prove that this
  obfuscation preserves the policy, i.e., that every program satisfying the
  policy still does after the transformation.
\end{abstract}

\section{Introduction}\label{sec:intro}
  \input{sections/intro}

\section{Background: CT-simulations}\label{sec:background}
  \input{sections/back}
\section{Proof of preservation}\label{sec:formalization}
  \input{sections/formalization}

\section{Conclusions}\label{sec:conclusions}
  \input{sections/concl}

\bibliographystyle{easychair}
\bibliography{biblio}

\appendix
\section{Proof}\label{sec:proof}
  \input{sections/proof}
\end{document}

%% file: sections/intro.tex

Secure compilation is an emerging research field that puts together techniques
  from security, programming languages, formal verification, and hardware
  architectures to devise compilation chains that protect various aspects of
  software and eliminate security vulnerabilities~\cite{sigplanblog,busi2019brief}.
As any other compiler, a secure one not only translates a \emph{source program},
written in a \emph{high-level language}, into an efficient \emph{object} code
(low-level), but also provides mitigations that make exploiting security
vulnerabilities more difficult and that limit the damage of an attack.
Moreover, a secure compilation chain deploys mechanisms to enforce secure
interoperability between code written in safe and unsafe languages, and
makes it hard extracting confidential data from information gained by examining
program runs (e.g., soft information as specific outputs given certain inputs,
or physical one as power consumption or execution time).

An important requirement for making a compiler secure is that it must grant that the security properties at the source level are fully preserved into the object level or, equivalently, that all the attacks that can be carried out at the object level can also be carried out at the source level.
In this way, it is enough showing that the program is secure at the source level, where reasoning is far more comfortable than at low level!

In this paper we focus on obfuscating compilers designed to protect a software
by obscuring its meaning and impeding the reconstruction of its original source code.
Usually, the main concern when defining such compilers is their robustness against
reverse engineering and the performance of the produced code.
Very few papers in the literature address the problem of proving their correctness, e.g.,~\cite{blazy2016formal}, and, to the best our knowledge, there is no
paper about the preservation of security policies.
Here, we offer a first contribution in this direction:
we consider a popular
program obfuscation (namely \emph{control-flow flattening}~\cite{laszlo2009obfuscating}) and a specific
security policy (namely \emph{constant-time}),
and we prove that every program satisfying the policy still does
after the transformation, i.e., the obfuscation preservers the policy.

For the sake of presentation, our source language is rather essential, as well as our illustrative examples.
The proof that control-flow flattening is indeed secure follows the approach
of~\cite{barthe2018secure} (briefly presented in Section~\ref{sec:background}), and only needs paper-and-pencil on our neat, foundational setting.
  Intuitively, we prove that if two executions of a program on different secret
  values are indistinguishable (i.e., they take the same time), then also the
  executions of its obfuscated version are indistinguishable (Section~\ref{sec:formalization}).

Actually, we claim that extending our results to a richer language will only
require to handle more details with no relevant changes in the structure of the
proof itself;
similarly, other security properties can be accommodated with no particular
effort in this framework, besides those already studied in~\cite{barthe2018secure}, and also
other program transformations can be proved to preserve security in the same
manner.
%
%

Below, we present the security policy and the transformation of interest.

\paragraph{Constant-time policy}
An intruder can extract confidential data by observing
the physical behavior of a system, through the so-called side-channel attacks.
The idea is that the attacker can recover some pieces of confidential information or can get indications on which parts are worth her cracking efforts, by measuring some physical quantity about the execution, e.g., power consumption and time.
Many of these attacks, called \emph{timing-based} attacks, exploit the execution time of programs~\cite{Kocher96}.
%
For example, if the program branches on a secret, the attacker may restrict the set of values it may assume, whenever the two branches have different execution times and the attacker can measure and compare them.
%
%
A toy example follows (in a sugared syntax),
where a user digits her pin then checked against the
stored one character by character: here the policy is violated since checking a correct pin takes longer than a wrong one.
%
%
%
\begin{lstlisting}
pin := read_secret();
current_char := 1;
while (current_char $\wleq$ stored_pin_length and pin(current_char) = stored_pin(current_char))
    current_char := current_char+1;

if (current_char = stored_pin_length+1) then print("OK!");
    else print("KO!");
    \end{lstlisting}
Many mitigations of timing-based attacks have been proposed, both hardware and software.
The \emph{program counter}~\cite{MolnarPSW05} and the \emph{constant-time}~\cite{Bernstein05} policies are software-based countermeasures,
giving rise to the \emph{constant-time programming} discipline.
It makes programs \emph{constant-time} w.r.t. secrets, i.e., the running times of programs is independent of secrets.
%
%
The requirement to achieve is that
neither the control-flow of programs nor the
sequence of memory accesses depend on secrets, e.g., the value of \texttt{pin} in our example.
Usually, this is formalized as a form of an information flow
policy~\cite{GM82} w.r.t.~an instrumented semantics that records information
leakage.
Intuitively, this policy requires that two executions started in equivalent
states (from an attacker’s point of view) yield equivalent leakage, making them \emph{indistinguishable} to an attacker.

The following is a constant-time version of the above program that checks if a pin is correct:
%
%
%
\begin{lstlisting}
pin := read_secret();
current_char := 1;
pin_ok := true
while (current_char $\wleq$ stored_pin_length)
    current_char := current_char+1
    if (pin(current_char) = stored_pin(current_char)) then pin_ok := pin_ok
         else pin_ok := false

if (pin_ok = true) then print("OK!");
    else print("KO!");
\end{lstlisting}

\paragraph{Control-flow flattening}
A different securing technique is code obfuscation, a program transformation that aims at hiding
the intention and the logic of programs by obscuring (portions of) source or object code.
It is used to protect a software making it more difficult to reverse engineer the (source/binary)
code of the program, to which the attacker can access.
%
%
%
%
In the literature different obfuscations have been proposed.
They range from only performing simple syntactic transformations, e.g., renaming variables
and functions, to more sophisticated ones that alter both the data, e.g., constant
encoding and array splitting~\cite{collberg2010surreptitious}, and the control flow of the
program, e.g., using opaque predicates~\cite{collberg2010surreptitious} and inserting dead code.

Control-flow flattening is an advanced obfuscation technique, implemented in state-of-the-art and industrial compilers, e.g.,~\cite{junod2015obfuscator}.
%
%
%
Intuitively, this transformation re-organizes the Control Flog Graph (CFG) of a program by taking its
basic blocks and putting them as cases of a selective structure that dispatches to the right case.
%
In practice, CFG flattening breaks each sequences of statements, nesting of
loops and if-statements into single statements, and then hides them in the cases
of a large \wswitch statement, in turn wrapped inside a \wwhile loop.
In this way, statements originally at different nesting level are now put next each
other.
Finally, to ensure that the control flow of the program during the execution is
the same as before, a new variable $\mtt{pc}$ is introduced that acts as a
program counter, and is also used to terminate the \wwhile loop.
The \wswitch statement dispatches the execution to one of its cases depending on
the value of $\mtt{pc}$.
When the execution of a case of the \wswitch statement is about to complete
$\mtt{pc}$ is updated with the value of the next statement to executed.

The obfuscated version of our constant-time example follows.
%
\begin{lstlisting}
$\mtt{pc}$ := 1;
while(1 $\leq \mtt{pc}$)
     switch($\mtt{pc}$):
          case 1:    pin := read_secret(); $\mtt{pc}$:= 2;
          case 2:    current_char := 1; $\mtt{pc}$:= 3;
          case 3:    pin_ok := true; $\mtt{pc}$:= 4;
          case 4:    if (current_char $\wleq$ stored_pin_length) then $\mtt{pc}$:= 5; else $\mtt{pc}$ := 9;
          case 5:    current_char := current_char+1; $\mtt{pc}$:= 6;
          case 6:    if (pin(current_char) = stored_pin(current_char)) then $\mtt{pc}$:= 7; else $\mtt{pc}$ := 8;
          case 7:    pin_ok := pin_ok; $\mtt{pc}$:= 4;
          case 8:    pin_ok := false; $\mtt{pc}$:= 4;
          case 9:    $\wskip$; $\mtt{pc}$:= 10;
          case 10: if (pin_ok = true) then $\mtt{pc}$:= 11; else $\mtt{pc}$ := 12;
          case 11: print("OK!"); $\mtt{pc}$:= 0;
          case 12: print("KO!"); $\mtt{pc}$:= 0;
\end{lstlisting}


  Now the point is whether the new obfuscated program is still constant-time, which is the case.
  In general we would like to have guarantees that the attacks prevented by the
  constant-time based countermeasure are not possible in the obfuscated
  versions.

%% file: sections/back.tex

Typically, for proving the correctness of a compiler one introduces a simulation relation between the computations at the source and at the target level: if such a relation exists, we have the guarantee that the source
program and the target program have the same observable behavior, i.e., the
same set of traces.

A general method for proving that constant-time is also preserved by
compilation generalizes this approach and is based on the notion of CT-simulation~\cite{barthe2018secure}.
It considers three relations: a simulation
relation between source and target, and two equivalences, one between source
and the other between target computations.
The idea is to prove that, given two computations at source level that are
equivalent, they are simulated by two equivalent computations at the target level.
Actually, CT-simulations guarantee the preservation of a
particular form of non-interference, called \emph{observational non-interference}.
In the rest of this section, we briefly survey observational
non-interference and how CT-simulations preserve it.

The idea is to model the behavior of programs using a
labeled transition system of the form $\step{\src{A}}{t}{\src{B}}$ where $\src{A}$ and $\src{B}$
are program configurations and $t$ represents the leakage associated with the
execution step between $\src{A}$ and $\src{B}$.
The semantics is assumed deterministic.
Hereafter, let the configurations of the source programs be
ranged over by $\src{A}, \src{B}, \ldots$ and those of the target programs be ranged over by $\trg{\alpha}, \trg{\beta}, \ldots$.
We will use the dot notation to refer to commands and state inside configurations, e.g., $\src{A.\mi{cmd}}$ refers to the command part of the configuration $\src{A}$.%
\footnote{Following the convention of secure compilation, we write in a \src{blue, sans\text{-}serif} font the elements of the source language, in a \trg{red, bold} one those of the target and in black those that are in common.}

The leakage represents what the attacker learns by the program execution.
Formally, the leakage is a list of atomic leakages where not
cancellable.
Observational non-interference is defined for complete executions (we denote
$S_f$ the set of final configurations) and w.r.t.\ an equivalence relation $\phi$ on configurations (e.g., states are
equivalent on public variables):
\begin{definition}[Observational non-interference~\cite{barthe2018secure}]
  A program $p$ is observationally non-interferent w.r.t.~a relation $\phi$,
  written $p \models \mathit{ONI}(\phi)$, iff
  for all initial configurations $A, A' \in S_i$ and configurations $B, B'$ and leakages $t,t'$ and $n \in
  \mathbb{N}$,
  \[
    \stepn{A}{t}{n}{B} \land \stepn{A'}{t'}{n}{B'} \land \phi(A, A') \implies
    t = t' \land (B \in S_f \text{ iff } B' \in S_f).
  \]
\end{definition}

Hereafter, we denote a compiler/transformation with $\comp{\cdot}$ and with $\comp{p}$
the result of compiling a program $\src{p}$.
Intuitively, a compiler $\comp{\cdot}$ preserves observational non-interference when for every
program $\src{p}$ that enjoys the property, $\comp{p}$ does as well.
Formally,
\begin{definition}[Secure compiler]\label{def:securecomp}
A transformation $\comp{\cdot}$ preserves observational non-interference \emph{iff}, for all programs $\src{p}$
\[
    \src{p} \models \mathit{ONI}(\phi) \Rightarrow \comp{p} \models \mathit{ONI}(\phi).
\]
\end{definition}
%
To show that a compiler $\comp{\cdot}$ is secure, we follow~\cite{barthe2018secure}, and build a \emph{general CT-simulation} in two steps.
First we define a simulation, called \emph{general simulation}, that relates computations between source and target languages.
The idea is to consider related a source and a target configuration whenever, after they perform a certain number of steps, they end up in two still related configurations.
Formally,
\begin{definition}[General simulation~\cite{barthe2018secure}]~\label{def:manystepssim}
    Let $\ns{\cdot}{\cdot}$ be a function mapping source and target configurations to $\mb{N}$.
    Also, let $\sz{\cdot}$ be a function from source configurations to $\mb{N}$.
    The relation $\confrel{}{\src{p}}{}$ is a \emph{general simulation} w.r.t.\ $\ns{\cdot}{\cdot}$ whenever:
    \begin{enumerate}
        \item $(\forall \src{B}, \trg{\alpha}.\ \step{\src{A}}{}{\src{B}} \land \confrel{\src{A}}{\src{p}}{\trg{\alpha}} \implies (\exists \trg{\beta}.\ \stepn{\trg{\alpha}}{}{\ns{\src{A}}{\trg{\alpha}}}{\trg{\beta}} \implies \confrel{\src{B}}{\src{p}}{\trg{\beta}})$,
        \item $(\forall \src{B}, \trg{\alpha}.\ \step{\src{A}}{}{\src{B}} \land \confrel{\src{A}}{\src{p}}{\trg{\alpha}} \land \ns{\src{A}}{\trg{\alpha}} = 0 \implies \sz{B} < \sz{A}$ \label{num:2},
        \item For any source configuration $\src{B} \in S_f$ and target configuration $\trg{\alpha}$ there exists a target configuration $\trg{\beta} \in S_f$ such that $\stepn{\trg{\alpha}}{}{\ns{\src{A}}{\trg{\alpha}}}{\trg{\beta}} \implies \confrel{\src{A}}{\src{p}}{\trg{\beta}}$.
    \end{enumerate}
\end{definition}
\noindent
Given two configurations $\src{A}$ and $\trg{\alpha}$ in the simulation relation, the function $\ns{A}{\alpha}$ predicts how many steps $\trg{\alpha}$ has to perform for reaching a target configuration $\trg{\beta}$ related with the corresponding source configuration $\src{B}$.
When $\ns{a}{\alpha} = 0$, a possibly infinite
sequence of source steps is simulated by an empty one at the target level.
To avoid these situations the measure function $\sz{\cdot}$ is introduced and the condition~\ref{num:2} of the above definition ensures that the measure of source configuration strictly decreases whenever 
\mbox{the corresponding target one stutters.}

The second step consists of introducing two equivalence relations between configurations: $\eqc_s$ relates configurations at the source and $\eqc_t$ at the target.
These two relations and the simulation relation form a \emph{general CT-simulation}.
Formally,
\begin{definition}[General CT-simulation~\cite{barthe2018secure}]
    A pair $(\eqc_s, \eqc_t)$ is a \emph{general CT-simulation} w.r.t.\ $\confrel{}{p}{}$, $\ns{\cdot}{\cdot}$ and $\sz{\cdot}$ whenever:
    \begin{enumerate}
        \item $(\eqc_s, \eqc_t)$ is a \emph{manysteps CT-diagram}, i.e., if
        \begin{itemize}
            \item $\src{A} \eqc_s \src{A'}$ and $\trg{\alpha} \eqc_t \trg{\alpha'}$;
            \item $\step{\src{A}}{t}{\src{B}}$ and $\step{\src{A'}}{t}{\src{B'}}$;
            \item $\stepn{\trg{\alpha}}{\tau}{\ns{A}{\alpha}}{\trg{\beta}}$ and $\stepn{\trg{\alpha'}}{\tau'}{\ns{A'}{\alpha'}}{\trg{\beta'}}$;
            \item $\confrel{A}{p}{\alpha}$, $\confrel{A'}{p}{\alpha'}$, $\confrel{B}{p}{\beta}$ and  $\confrel{B'}{p}{\beta'}$
        \end{itemize}
        then
        \begin{itemize}
            \item $\tau = \tau'$ and $\ns{A}{\alpha} = \ns{A'}{\alpha'}$;
            \item $\src{B} \eqc_s \src{B'}$ and $\trg{\beta} \eqc_t \trg{\beta'}$;
        \end{itemize}
        \item if $\src{A}, \src{A'}$ are initial configurations, with targets $\trg{\alpha}, \trg{\alpha'}$, and $\phi(\src{A}, \src{A'})$, then $\src{A} \eqc_s \src{A'}$ and  $\trg{\alpha} \eqc_t \trg{\alpha'}$;
        \item If $\src{A} \eqc_s \src{A'}$, then $\src{A} \in S_f \iff \src{A'} \in S_f$;
        \item $(\eqc_s, \eqc_t)$ is a \emph{final CT-diagram}~\cite{barthe2018secure}, i.e., if
        \begin{itemize}
            \item $\src{A} \eqc_s \src{A'}$ and $\trg{\alpha} \eqc_t \trg{\alpha'}$;
            \item $\src{A}$ and $\src{A'}$ are final;
            \item $\stepn{\trg{\alpha}}{\tau}{\ns{A}{\alpha}}{\trg{\beta}}$ and $\stepn{\trg{\alpha'}}{\tau'}{\ns{A'}{\alpha'}}{\trg{\beta'}}$;
            \item $\confrel{A}{p}{\alpha}$, $\confrel{A'}{p}{\alpha'}$, $\confrel{B}{p}{\beta}$ and $\confrel{B'}{p}{\beta'}$
        \end{itemize}
        then
        \begin{itemize}
            \item $\tau = \tau'$ and $\ns{A}{\alpha} = \ns{A'}{\alpha'}$;
            \item $\trg{\beta} \eqc_t \trg{\beta'}$ and they are both final.
        \end{itemize}
     \end{enumerate}
\end{definition}
The idea is that the relations $\eqc_s$ and $\eqc_t$ are stable under reduction, i.e., preservation of the observational non-interference is guaranteed.
The following theorem, referred to in~\cite{barthe2018secure} as Theorem 6, gives a sufficient condition to establish constant-time preservation.
\begin{theorem}[Security]\label{th:security}
If $\src{p}$ is constant-time w.r.t. $\phi$ and there is a general CT-simulation w.r.t.\ a general simulation, then $\comp{p}$ is constant-time w.r.t. $\phi$.
\end{theorem}

%% file: sections/formalization.tex

In this section, we present the proof that control-flow flattening preserves
constant-time policy.
We first introduce a small imperative language, its semantics in
the form of a LTS and our leakage model.
Then, we formalize our obfuscation as a function from syntax to syntax, and
finally we prove the preservation of the security policy.

%
%

\subsection{The language and its (instrumented) semantics}
We consider a small imperative language with arithmetic and boolean expressions.
Let $\mi{Var}$ be a set program identifiers, the syntax is
\begin{align*}
    AExpr \ni e &\Coloneqq v \mid \wvar{x} \mid e_1 \wop{op} e_2
    \qquad \qquad v \in \mb{Z}, \quad \wop{op} \in \{\wop{+}, \wop{-}, \wop{*}, \wop{/}, \wop{\%}\}, \quad \wvar{x} \in\mi{Var} \\
    BExpr \ni b &\Coloneqq \mtt{true} \mid \mtt{false} \mid b_1 \wop{or} b_2 \mid \wop{not} b \mid e_1 \wleq e_2 \mid  e_1 \weq e_2\\
    Cmd \ni c &\Coloneqq \wskip \mid \wvar{x} \wassign e \mid c_1 \wseq c_2 \mid \wif b \wthen c_1 \welse c_2 \mid \wwhile b \wdo c
\end{align*}
We assume that each command in the syntax carries a permanent color either \emph{white} or not, typically $\digamma$.
Also, we stipulate that each \wwhile statement and all its components get a
unique non-white color, and that there is a function $\mi{color}$ yielding the
color of a statement.

%


Now,  we define the semantics and instantiate the framework of~\cite{barthe2018secure} to the \emph{non-cancelling constant-time policy}.
For that, we define a \emph{leakage model} to describe the information that an attacker can observe during the execution.
Recall from the previous section that the leakage is a list of atomic leaks.
We denote with $\cdot$ the list concatenation and with $[a]$ a list with a single element $a$.
Arithmetic and boolean expressions leak the sequence of operations required to
be evaluated; we assume that there is an observable $\underbar{op}$,
associated with the arithmetic operation being executed, but not with the logical ones (slightly simplifying
~\cite{barthe2018secure}).
Also we denote with $\eleak$ absence of leaking.
Our leakage model is defined by the following function $\leak{\cdot}{\cdot}$ that given
an expression (either arithmetic or boolean) and a state returns the
corresponding leakage:
%
\begin{align*}
    \leak{v}{\sigma} &= \leak{\wvar{x}}{\sigma} = \leak{\mtt{true}}{\sigma} = \leak{\mtt{false}}{\sigma} = [\eleak]\\
    \leak{\wop{not} b}{\sigma} &= \leak{b}{\sigma}\\
    \leak{e_1 \wop{op} e_2}{\sigma} &= \leak{e_1}{\sigma} \wlc \leak{e_2}{\sigma} \wlc \underbar{op}\\
    \leak{e_1 \wleq e_2}{\sigma} &= \leak{e_1 \weq e_2}{\sigma} = \leak{b_1 \wop{or} b_2}{\sigma} =\leak{e_1}{\sigma} \wlc \leak{e_2}{\sigma}
\end{align*}
Accesses to constants and identifiers leak nothing; boolean and relational expressions leak the concatenation of the leaks of their sub-expressions;
the arithmetic expressions append the observable of the applied operator to the leaks of their
sub-expressions.

We omit the semantics of arithmetic and boolean expression $\sem{\cdot}$ because fully standard~\cite{nielson2007semantics}; we only assume that each syntactic arithmetic operator
$\wop{op}$ has a corresponding semantic \mbox{operator $\mi{op}$.}

The semantics of commands is given in term of a transition relation $\xrightarrow{t}$ between configurations where $t$ is the leakage of that transition step.
As usual a configuration is a pair $c, \sigma$ consisting of a command and a state $\sigma \in \mi{Store}$ assigning values to program identifiers.
Given a program $p$ the set of initial configurations is $S_i = \{ p, \sigma \mid \sigma \in \mi{Store} \}$, and that of final configurations is $S_f = \{ \wskip, \sigma \mid \sigma \in \mi{Store} \}$.

Figure~\ref{fig:cmdsem} reports the instrumented semantics of the language.
Moreover, the semantics is assumed to keep colors, in particular in the rule for an $\digamma$-colored $\wwhile\!\!$, all the components of the $\wif\!\!$ in the target are also $\digamma$-colored, avoiding color clashes (see the .pdf for colors).
\begin{figure}
    \footnotesize
    \begin{mathpar}
        \inferrule
        { }
        {\step{\wvar{x} \wassign e, \sigma}{\leak{e}{\sigma} \wlc [\wvar{x}]}{\wskip, \sigma\{ \wvar{x} \mapsto \sem{a}\}}}

        \inferrule
        {
            \step{c_1, \sigma}{t}{c'_1, \sigma'}
        }
        {
            \step{c_1 \wseq c_2, \sigma}{t}{c'_1 \wseq c_2, \sigma'}
        }

        \inferrule
        { }
        {
            \step{\wskip \wseq c_2, \sigma}{t}{c_2, \sigma'}
        }

        \inferrule
        {
            \sem{b} = \mi{true}
        }
        {
            \step{\wif b \wthen c_1 \welse c_2, \sigma}{\leak{b}{\sigma}\wlc[\mi{true}]}{c_1, \sigma}
        }

        \inferrule
        {
            \sem{b} = \mi{false}
        }
        {
            \step{\wif b \wthen c_1 \welse c_2, \sigma}{\leak{b}{\sigma}\wlc[\mi{false}]}{c_2, \sigma}
        }


        \inferrule
        { }
        {
            \step{\color{ForestGreen}{\wwhile b \wdo c}, \sigma}{[\eleak]}{\color{ForestGreen}{\wif b \wthen ( c \wseq \wwhile b \wdo c ) \welse \wskip}, \sigma}
        }



    \end{mathpar}

    \caption{Instrumented operational semantics for commands.} \label{fig:cmdsem}
\end{figure}
%


\subsection{Control-flow flattening formalization}\label{sec:obfu}

Recall that the initial program being obfuscated is $\src{p}$.
For the sake of presentation, we will adopt the sugared syntax we used in Section~\ref{sec:intro} and represent a sequence of nested conditionals
in the obfuscated program as the command $\trg{\wswitch e : cs}$, where $\trg{cs} = \trg{\wlist{(v_1, c_1)}{\ldots}{(v_n : c_n)}}$, with semantics
\begin{mathpar}
{\footnotesize
    \inferrule
    {
        (\sem{\trg{e}}, \trg{c}) \notin \trg{\mtt{cs}}
    }
    {
        \step{\trg{\wswitch e : \mtt{cs}, \sigma}}{\leak{e}{\sigma}}{\trg{\wskip, \sigma}}
    }

    \inferrule
    {
        (\sem{\trg{e}}, \trg{c}) \in \trg{\mtt{cs}}
    }
    {
        \step{\trg{\wswitch e : \mtt{cs}, \sigma}}{\leak{e}{\sigma}}{\trg{c, \sigma}}
    }
}
\end{mathpar}

\noindent
Now, let $\trg{\wvar{pc}}$ be a fresh identifier, called \emph{program counter}.
Then, following~\cite{blazy2016formal},%
the obfuscated version $\comp{c}$ of the command $\src{c}$ is
%
  \vspace{-2mm}
\[
\trg{
\begin{aligned}
    &\trg{\wvar{pc}} \wassign 1 \wseq\\
                 &\wwhile 1 \wleq \trg{\wvar{pc}} \wdo\\
                 &\qquad \wswitch \trg{\wvar{pc}} : \com{\cmdlbl{\trg{\wvar{pc}}}{c}{1}{0}}
\end{aligned}}
\]
%
where
\vspace{-2mm}
%
\begin{align*}
&    \cmdlbl{\trg{\wvar{pc}}}{\wskip}{n}{m} = \wlist{(n, \trg{\wskip \wseq \trg{\wvar{pc}} \wassign \com{m}})}\\
&    \cmdlbl{\trg{\wvar{pc}}}{x \wassign e}{n}{m} = \wlist{(n, \trg{x \wassign e \wseq \trg{\wvar{pc}} \wassign \com{m}})}\\
&    \cmdlbl{\trg{\wvar{pc}}}{c_1 \wseq c_2}{n}{m} = \cmdlbl{\trg{\wvar{pc}}}{c_1}{n}{n+\size{c_1}} \wlc \cmdlbl{\trg{\wvar{pc}}}{c_2}{n+\size{c_1}}{m}\\
&    \cmdlbl{\trg{\wvar{pc}}}{\wif b \wthen c_1 \welse c_2}{n}{m} = \\
&   \qquad\qquad \wlist{(n, \trg{\wif b \wthen \trg{\wvar{pc}} \wassign \com{n + 1} \welse \trg{\wvar{pc}}
\wassign \com{n+1+\size{c_1}}})} \wlc\\
&       \qquad \qquad  \cmdlbl{\trg{\wvar{pc}}}{c_1}{n+1}{m} \wlc 
 \cmdlbl{\trg{\wvar{pc}}}{c_2}{n+1+\size{c_1}}{m}\\
&  \cmdlbl{\trg{\wvar{pc}}}{\wwhile b \wdo c}{n}{m} = \\
& \qquad \qquad \wlist{(n, \trg{\wif b \wthen \trg{\wvar{pc}} \wassign \com{n+1} \welse \trg{\wvar{pc}}
\wassign \com{n+1+\size{c}}})} \wlc\\
&  \qquad  \qquad \cmdlbl{\trg{\wvar{pc}}}{c}{n+1}{n}  \wlc \wlist{(n+1+\size{c}, \trg{\wskip \wseq \trg{\wvar{pc}} \wassign \com{m}})}
\end{align*}
with $\size{\cdot}$ defined as follows
\vspace{-3mm}
\begin{align*}
    \size{c} &= 1 \text{ if } c \in \{ \wskip, \cdot \wassign \cdot \}\\
    \size{c_1 \wseq c_2} &= \size{c_1} + \size{c_2}\\
    \size{\wif b \wthen c_1 \welse c_2} &= 1 + \size{c_1} + \size{c_2}\\
    \size{\wwhile b \wdo c} &= 2 + \size{c}
\end{align*}
The obfuscated version of a program  $\src{p}$ is a
loop with condition $\trg{1 \leq \trg{\wvar{pc}}}$ and with body a $\trg{\wswitch}$
statement.
The $\trg{\wswitch}$ condition is on the values of $\trg{\wvar{pc}}$ and its cases
correspond to the flattened statements, obtained from the function
$\cmdlbl{\trg{\wvar{pc}}}{c}{n}{m}$.
It returns a list containing the cases of the $\trg{\wswitch}$ and it is inductively defined on the syntax of commands:
the first parameter $\trg{\wvar{pc}}$ is the identifier to use for program counter; the second is
the command $\src{c}$ to be flattened; the parameter $n$ represents the value of the
guard of the case generated for the first statement of $\src{c}$; the last parameter
$m$ represents the value to be assigned to $\trg{\wvar{pc}}$ by the last $\trg{\wswitch}$
case generated.
For example, the flattening of a sequence
$\src{c_1 \wseq c_2}$ generates the cases corresponding to $\src{c_1}$ and
$\src{c_2}$, and then concatenates them.
Note that the values of the program counter for the cases of $\src{c_2}$ start from the value
assigned to $\trg{\wvar{pc}}$ by the last case generated for $\src{c_1}$, i.e., $n +
\size{c_1}$, where the function $\size{\cdot}$ returns the ``length'' of $\src{c_1}$.
For a program $\src{p}$, we use $1$ as initial value of $n$ and $0$ as last value to be assigned so as to exit from the $\trg{\wwhile}$ loop.

  \subsection{Correctness and security}

Since obfuscation does not change the language (apart from sugaring nested $\wif\!\!$ commands); the operational semantics is deterministic; and there are no \emph{unsafe} programs (i.e., a program gets stuck iff execution has completed), the correctness of obfuscation directly follows from the existence of a general simulation between the source and the target languages~\cite{barthe2018secure}.
For that, inspired by~\cite{blazy2016formal}, we define the relation $\confrel{}{p}{}$
between source and target configurations shown in \figurename~\ref{fig:simobf}.
Intuitively, the relation $\confrel{}{p}{}$ matches source and target
configurations with the same behaviour, depending on whether they are final
(third rule), their execution originated from a loop
(Rule~\rulename{(Colored)}) or not (Rule~\rulename{(White)}).
Note that we differentiate white and colored cases as to avoid circular
reasoning in the derivations of $\confrel{}{p}{}$.
More specifically, our relation matches a configuration $\src{A}$ in the source
with a corresponding $\trg{\alpha}$ in the target.
Actually, $\trg{\alpha.\mi{cmd}}$ is the $\trg{\wwhile}$ loop of the obfuscated program (fourth premise in Rule~\rulename{(White)} and third in Rule~\rulename{(Colored)}), whereas $\trg{\alpha.\sigma}$ is equal to $\src{A.\sigma}$ except for the value of $\trg{\wvar{pc}}$.
Its value is mapped to the case of the $\trg{\wswitch}$ corresponding to the next command in $\src{A}$ (first premise in Rule~\rulename{(White)} and fifth in
Rule~\rulename{(Colored)}).

To understand how our simulation works, recall the example from
Section~\ref{sec:intro}.
By Rule~\rulename{(White)} we relate the configuration reached at line~$(3)$
at the source level with that of the obfuscated program starting at
line~$(2)$ and with a state equal to that of the source level with the additional binding $\trg{\wvar{pc}} \mapsto \trg{3}$.
Similarly, we relate the configuration reached at line~$(6)$ at the source
level and its obfuscated counterpart (again at line~$(2)$ at the obfuscated
level), using Rule~\rulename{(Colored)} and noting that the source
configuration derives from the execution of a loop.

The following theorem  ensures that the relation $\confrel{}{p}{}$ is a general simulation.
\begin{restatable}{theorem}{thmgensim}\label{thm:gensim}
For all programs $\src{p}$, the relation $\confrel{}{p}{}$ is a general simulation.
\end{restatable}

The correctness of the obfuscation is now a corollary of Theorem~\ref{thm:gensim}.
%
\begin{corollary}[Correctness]\label{thm:correctness}
For all commands $\src{c}$ and store $\sigma$
\[
\src{c}, \sigma \rightarrow^* \src{\wskip}, \sigma' \quad \text{iff} \quad \trg{\comp{c}}, \sigma \rightarrow^* \trg{\wskip}, \sigma'
\]
\end{corollary}

The next step is showing that the control-flow flattening obfuscation preserves the constant-time programming policy.
For that we define $\eqc$ below and we show that $(\eqc, \eqc)$ is a general CT-simulation, as required by Theorem~\ref{th:security}.
\begin{definition}\label{def:eqc}
    Let $\mathcal A$ and $\mathcal{A}'$ be two (source or obfuscated) configurations, then $\mathcal{A} \eqc \mathcal{A}'$ iff $\mathcal{A}.\mi{cmd} = \mathcal{A}'.\mi{cmd}$. 
\end{definition}
We prove the following:
\begin{restatable}{theorem}{thmctsim}\label{thm:ctsim}
The pair $(\eqc, \eqc)$ is a \emph{general CT-simulation} w.r.t.\ $\confrel{}{p}{}$, $\ns{\cdot}{\cdot}$ and $\sz{\cdot}$.
\end{restatable}
The main result of our paper directly follows from the theorem above, because the transformation in Section~\ref{sec:obfu} satisfies Definition~\ref{def:securecomp}:
\begin{corollary}[Constant-time preservation]\label{thm:ctpres}\ \\
The control-flow fattening obfuscation preserves the constant-time policy.
\end{corollary}

The proofs of the theorems above are in the Appendix~\ref{sec:proof}.

\begin{figure}[tb]
\footnotesize
\begin{mathpar}
    \inferrule
    [{(White)}]
    {
        \mi{color}(\src{c}) = \mi{white}\\
        \trg{\sigma'} = \src{\sigma} \cup \{ \trg{\wvar{pc}} \mapsto n\} \\
        \trg{\mtt{ls}} = \cmdlbl{\trg{\wvar{pc}}}{p}{1}{0}\\
        \trg{c'} = \trg{\wwhile (1 \wleq \trg{\wvar{pc}}) \wdo (\wswitch \trg{\wvar{pc}} : \mtt{ls})}\\
        \cmdrelc{-,\trg{\wvar{pc}}}{\mi{c}}{\trg{\mtt{ls}}[n], m}
    }
    {
        \confrel{c, \sigma}{p}{c', \sigma'}
    }

    \inferrule
    [{(Colored)}]
    {
        \trg{\sigma'} = \src{\sigma} \cup \{ \trg{\wvar{pc}} \mapsto n \} \quad
        \trg{\mtt{ls}} = \cmdlbl{\trg{\wvar{pc}}}{p}{1}{0}\\
        \trg{c'} = \trg{\wwhile (1 \wleq \trg{\wvar{pc}}) \wdo (\wswitch \trg{\wvar{pc}} : \mtt{ls})}\\
        \src{\wwhile b \wdo c''} \in \src{p}\\
        \mi{color}(\src{\wwhile b \wdo c''}) = \mi{color}(\src{c}) \neq \mi{white}\\
        \cmdrelc{-,\trg{\wvar{pc}}}{\wwhile b \wdo c''}{\trg{\mtt{ls}}[n_0], m'}\\
        \cmdrelw{n_0, \trg{\wvar{pc}}}{c}{\trg{\mtt{ls}}[n], m} \\
    }
    {
        \confrel{c, \sigma}{p}{c', \sigma'}
    }

    \inferrule
    {
        \trg{\sigma'} = \src{\sigma} \cup \{ \trg{\wvar{pc}} \mapsto n \}
    }
    {
        \confrel{\wskip, \sigma}{p}{\wskip, \sigma'}
    }

    \\\mathclap{\rule{\textwidth}{0.4pt}}\\

    \inferrule
    { }
    {
        \cmdrel{n_0, \trg{\wvar{pc}}}{\wskip}{\trg{\mtt{ls}}[n], m}
    }



    \inferrule
    {
        \trg{\mtt{ls}}[n] = (n, \trg{\wvar{x} \wassign e \wseq \wvar{pc} \wassign \com{m}})
    }
    {
        \cmdrel{n_0, \trg{\wvar{pc}}}{\wvar{x} \wassign e}{\trg{\mtt{ls}}[n], m}
    }


    \inferrule
    {
        \cmdrel{n_0, \trg{\wvar{pc}}}{c_1}{\trg{\mtt{ls}}[n],m'} \\
        \cmdrel{n_0, \trg{\wvar{pc}}}{c_2}{\trg{\mtt{ls}}[m'],m}
    }
    {
        \cmdrel{n_0, \trg{\wvar{pc}}}{c_1 \wseq c_2}{\trg{\mtt{ls}}[n],m}
    }

    \inferrule
    {
        \trg{\mtt{ls}}[n] = (n, \trg{\wif b \wthen \wvar{pc} \wassign \com{n+1} \welse \trg{\wvar{pc}} \wassign \com{n+1+\size{c_1}}})\\
        \cmdrel{n_0, \trg{\wvar{pc}}}{c_1}{\trg{\mtt{ls}}[n+1], m}\\
        \cmdrel{n_0, \trg{\wvar{pc}}}{c_2}{\trg{\mtt{ls}}[n+1+\size{c_1}], m}
    }
    {
        \cmdrel{n_0, \trg{\wvar{pc}}}{\wif b \wthen c_1 \welse c_2}{\trg{\mtt{ls}}[n], m}
    }

    \inferrule
    {
        \trg{\mtt{ls}}[n] = (n, \trg{\wskip \wseq \wvar{pc} \wassign \com{n_0}}c)\\
    }
    {
        \cmdrelw{n_0, \trg{\wvar{pc}}}{\wwhile b \wdo c}{\trg{\mtt{ls}}[n], n_0}
    }

    \inferrule
    { }
    {
        \cmdrelw{n_0, \trg{\wvar{pc}}}{\wwhile b \wdo c}{\trg{\mtt{ls}}[n_0], n_0}
    }

    \inferrule
    {
        \cmdrelw{n, \trg{\wvar{pc}}}{\wif b \wthen (c \wseq \wwhile b \wdo c) \welse \wskip}{\trg{\mtt{ls}}[n], m}
    }
    {
        \cmdrelc{- , \trg{\wvar{pc}}}{\wwhile b \wdo c}{\trg{\mtt{ls}}[n], m}
    }
\end{mathpar}

where $\sim\, \in \{ \diamond, \bowtie \}$, and the first parameter ($n_0$) is immaterial in $\bowtie$.

\caption{Definition of $\confrel{}{p}{}$ relation on configurations and its auxiliary relations.} \label{fig:simobf}
\end{figure}


%% file: sections/concl.tex

In this paper we applied a methodology from the literature~\cite{barthe2018secure} to the advanced obfuscation technique of control-flow flattening and proved that it preserves the constant-time policy.
For that, we have first defined what programs leak.
Then, we have defined the relation $\confrel{}{p}{}$ between source and target configurations -- that roughly relates configurations with the same behavior -- and proved that it adheres to the definition of general simulation.
Finally, we proved that the obfuscation preserves constant time by showing that the pair $(\eqc, \eqc)$ is a general CT-simulation, as required by the framework we instantiated.
As a consequence, the obfuscation based on control-flow flattening is proved to preserve the constant-time policy.

Future work will address proving the security of other obfuscations techniques, and considering other security properties, e.g., general safeties or hyper-safeties.
Here we just considered a passive attacker that can only observe the leakage, and an interesting problem would be to explore if our result and the current proof technique scale to a setting with active attackers that also interferes with the execution of programs.
Indeed, recently new secure compilation principles have been proposed to take active attackers into account~\cite{abate2018exploring}.

\paragraph{Related Work}

Program obfuscations are widespread code
transformations~\cite{laszlo2009obfuscating,collberg2010surreptitious,junod2015obfuscator,tigress,uglifyjs2,binaryen}
designed to protect software in settings where the adversary has physical
access to the program and can compromise it by inspection or tampering.
A great deal of work has been done on obfuscations that are resistant
against reverse engineering making the life of attackers harder.
However, we do not discuss these papers because they do not
consider formal properties of the proposed transformations.
We refer the interested reader to~\cite{hosseinzadeh2018} for a recent survey.

Since to the best our knowledge, ours is the first work addressing the problem
of security preservation, here we focus only on those proposals that formally
studied the correctness of obfuscations.
%
In~\cite{dallapreda2005control, dallapreda2009semantics} a formal framework
based on abstract interpretation is proposed to study the effectiveness of
obfuscating techniques.
This framework not only characterizes when a transformation is correct but also
measures its resilience, i.e., the difficulty of undoing the obfuscation.
More recently, other work went in the direction of fully verified, obfuscating
compilation chains~\cite{blazy2012towards, blazy2016formal,blazy2019formal}.
Among these~\cite{blazy2016formal} is the most similar to ours, but it only focusses
on the correctness of the transformation, and studies it in the setting of the
CompCert C compiler.
Differently, here we adopted a more foundational approach by considering a core
imperative language and proved that the considered transformation preserves security.

As for secure compilation, we can essentially distinguish two different approaches.
The first one only considers passive attackers (as we do) that do not interact
with the program but that try to extract confidential data by observing its behaviour.
Besides~\cite{barthe2018secure}, recently there has been an increasing interest in
preserving the verification of the constant time policy, e.g.,
a version of the CompCert C compiler~\cite{barthe2020formal} has been
released that guarantees that preservation of the policy in each compilation step.
The second approach in secure compilation considers active attackers that are
modeled as contexts in which a program is plugged in.
Traditionally, this approach reduces proving the security preservation to
proving that the compiler is fully-abstract~\cite{patrignani2019formal}.
However, recently new proof principles emerged,
see~\cite{abate2018exploring,patrignani2018robustly} for an overview.
%


%% file: sections/proof.tex

Here we report a proof sketch with that includes the most significant cases.



Before proving the correctness of the obfuscation, we prove the following lemma that relates the termination of the source program with the assignment to the special variable $\wvar{pc}$ causing the termination of the obfuscated version.
(As above, the $n$-th element of the list $\mtt{ls}$ generated by the obfuscation is referred to as $\mtt{ls}[n]$.)
\begin{lemma}\label{lemma:skip-implies-pc0}
    Let $p$ be a program with $\size{p} = n$, and $\mtt{ls} = \cmdlbl{\wvar{pc}}{p}{1}{0}$.
    \\
    If $\stepn{p, \sigma_\mi{init}}{}{*}{\step{c, \sigma}{}{\wskip, \sigma'}}$ then $\mtt{ls}[n] = (n, c \wseq \wvar{pc} \wassign 0)$.
\end{lemma}
\begin{proof}
    Easily proved by induction on $\size{p}$.
\end{proof}

The binary relation $\confrel{}{p}{}$ in Figure~\ref{fig:simobf} is a general simulation according to Definition~\ref{def:manystepssim}.
Before showing that, we first give the definition of $\ns{A}{\alpha}$, that maps a source $\src{A}$ and a target configuration $\alpha$ into a natural number expressing the number of steps that need to be performed on $\alpha$ to reach a configuration which is in relation $\confrel{}{p}{}$ with the one reached from $\src{A}$ in one step.~\cite{barthe2018secure}
In our case the definition is syntax directed and is as follows: 
\begin{align*}
    \ns{A}{\alpha} \triangleq
        \begin{cases}
            0 & \text{if } A.\mi{cmd} \in \{ \wskip \wseq \cdot, \wwhile \cdot \wdo \cdot \} \\
            \ns{(c_1, A.\sigma)}{\alpha} & \text{if } A.\mi{cmd} = c_1 \wseq c_2 \land c_1 \neq \wskip\\
            9 & \text{if } A.\mi{cmd} \in \{ \wif \cdot \wthen \cdot \welse \}\\
            8 & \text{o.w.}
        \end{cases}
\end{align*}
We also define the measure $\sz{\cdot}$ that it used to guarantee that an infinite number of steps at the source level is not matched by a finite number of steps at the target~\cite{barthe2018secure,blazy2016formal}: 
\begin{align*}
    | A | \triangleq
        \begin{cases}
            2\cdot| (c, A.\sigma) | + 3 & \text{ if } A.\mi{cmd} = \wwhile b \wdo c\\
            | (c_2, A.\sigma) | + 1 & \text{ if } A.\mi{cmd} = \wskip \wseq c_2\\
            0 & \text{ o.w.}
        \end{cases}
\end{align*}
Note that the measure $\sz{\cdot}$ was built with the specific requirements of the proof of Theorem~\ref{thm:gensim}, i.e., to ensure that $(\forall B, \alpha.\ \step{A}{}{B} \land \confrel{A}{p}{\alpha} \land \ns{A}{\alpha} = 0 \implies \sz{B} < \sz{A})$.

\thmgensim*

\begin{proof} (Sketch)

    \begin{itemize}
        \item \underline{$(\forall B, \alpha.\ \step{A}{}{B} \land \confrel{A}{p}{\alpha} \implies (\exists \beta.\ \stepn{\alpha}{}{\ns{A}{\alpha}}{\beta} \implies \confrel{B}{p}{\beta})$.}
        This proof goes by induction on the rules of the operational semantics.
        We only consider the most interesting cases, the others being similar.
        Actually, we consider two base cases and the only inductive one.

        Also, note that -- by definition of $\confrel{}{p}{}$ -- any configuration $\alpha$ related with another $\src{A}$ must be such that $\alpha.\sigma = A.\sigma \cup \{ \wvar{pc} \mapsto n \}$ and
        \begin{align*}
            \alpha.\mi{cmd} = &\wwhile 1 \wleq \wvar{pc} \wdo \\
                     &\qquad \wswitch \mtt{pc} : \mtt{ls}
        \end{align*}
        for some $n$ and $\mtt{ls} = \cmdlbl{\wvar{pc}}{p}{1}{0}$.

        \begin{description}
            \item[Case: $A.\mi{cmd} = \wvar{x} \wassign e$.]
            By definition of $\rightarrow$ we know that:
            \[
                \step{A}{}{b = \wskip, \sigma[\wvar{x} \mapsto \sem{e}]}.
            \]

            We have two exhaustive cases, depending on $\mi{color}(A.\mi{cmd})$:
            \begin{enumerate}
                \item \textbf{Case $\mi{color}(A.\mi{cmd}) = \mi{white}$.}
                    By Rule~\rulename{(White)} we know that $\cmdrelc{-, \wvar{pc}}{\wvar{x} \wassign e}{\mtt{ls}[n]}$.

                    By Lemma~\ref{lemma:skip-implies-pc0} and definition of $\bowtie$, we know that
                    \mbox{$\mtt{ls}[n] = (n, \wvar{x} \wassign e \wseq \wvar{pc} \wassign 0)$}
                    and $\stepn{\alpha}{}{\ns{A}{\alpha}}{\beta = \wskip, \sigma[\wvar{x} \mapsto \sem{e}, \wvar{pc} \mapsto 0]}$. 
                    The thesis then follows by definition of $\confrel{}{p}{}$.
                \item \textbf{Case $\mi{color}(A.\mi{cmd}) \neq \mi{white}$.}
                    Similarly to the case above, by Rule~\rulename{(Colored)} we have $\cmdrelc{-,\wvar{pc}}{\wwhile b \wdo c''}{\mtt{ls}[n_0]}$ and $\cmdrelw{n_0, \wvar{pc}}{\wvar{x} \wassign e}{\mtt{ls}[n]}$.

                    By Lemma~\ref{lemma:skip-implies-pc0} and definition of $\diamond$, we know that $\mtt{ls}[n] = (n, \wvar{x} \wassign e \wseq \wvar{pc} \wassign 0)$, thus $\stepn{\alpha}{}{\ns{A}{\alpha}}{\beta = \wskip, \sigma[\wvar{x} \mapsto \sem{e}, \wvar{pc} \mapsto 0]}$. 
                    The thesis then follows by definition of $\confrel{}{p}{}$.
            \end{enumerate}
            \item[Case: $A.\mi{cmd} = \wwhile b \wdo c$.]
            By definition of $\rightarrow$ we know that:
            \[
                \step{A}{}{B = \wif b \wthen (c \wseq \wwhile b \wdo c) \welse \wskip, \sigma[\wvar{x} \mapsto \sem{e}]}.
            \]

            Again, we have two exhaustive cases, depending on $\mi{color}(A.\mi{cmd})$:
            \begin{enumerate}
                \item \textbf{Case $\mi{color}(A.\mi{cmd}) = \mi{white}$.}
                    By Rule~\rulename{(White)} of $\confrel{}{p}{}$ we know that $\cmdrelc{-, \wvar{pc}}{\wwhile b \wdo c}{\mtt{ls}[n], m}$, i.e. that $(\star)\  \cmdrelw{n, \wvar{pc}}{\wif b \wthen (c \wseq \wwhile b \wdo c) \welse \wskip}{\mtt{ls}[n], m}$.

                    Since $\ns{A}{\alpha} = 0$, $\beta = \alpha$.
                    We must then show that $\confrel{b}{p}{\beta}$ and, given that $\mi{color}(\beta.\mi{cmd}) \neq \mi{white}$ since it derivates from $\wwhile b \wdo c$, it suffices to show the following facts
                    \begin{itemize}
                        \item $\beta.\sigma = B.\sigma \cup \{ \wvar{pc} \mapsto n \}$ and $\beta.\mi{cmd} = \wwhile 1 \wleq \wvar{pc} \wdo (\wswitch \wvar{pc} : \mtt{ls})$ with $\mtt{ls} = \cmdlbl{\wvar{pc}}{p}{1}{0}$ that directly follows from $\beta = \alpha$;
                        \item $\cmdrelc{-, \wvar{pc}}{\wwhile b \wdo c}{\mtt{ls}[n_0], m'}$.
The thesis follows from $\star$ and by definition of $\diamond$ since $n_0 = n$, because $\alpha.\sigma(\wvar{pc}) = n$, and by choosing $m' = m$;

                        \item $\cmdrelw{n_0, \wvar{pc}}{\wif b \wthen (c \wseq \wwhile b \wdo c) \welse \wskip}{\mtt{ls}[n], m}$ directly follows from the hypotheses.
                    \end{itemize}

                \item \textbf{Case $\mi{color}(A.\mi{cmd}) \neq \mi{white}$.}
                    Analogous to the above.
            \end{enumerate}

            \item[Case: $A.\mi{cmd} = c_1 \wseq c_2$, $c_1 \neq \wskip$.]
                The induction hypothesis (IHP) reads as follows
                \begin{align*}
                    \forall \alpha'.\ \step{c_1, \sigma}{}{c'_1, \sigma'} \land \confrel{c_1, \sigma}{p}{\alpha'}
                        \Rightarrow (\exists \beta'.\ \stepn{\alpha'}{}{\ns{(c_1, \sigma)}{\alpha'}}{\beta'} \Rightarrow \confrel{c'_1, \sigma'}{p}{\beta'})
                \end{align*}
                    and we have to prove that
                \begin{align*}
                    \forall \alpha.\ \step{c_1 \wseq c_2, \sigma}{}{c'_1 \wseq c_2, \sigma'} &\land \confrel{c_1 \wseq c_2, \sigma}{p}{\alpha} \\
                        &\Rightarrow (\exists \beta.\ \stepn{\alpha}{}{\ns{(c_1 \wseq c_2, \sigma)}{\alpha}}{\beta} \Rightarrow \confrel{c'_1 \wseq c_2, \sigma}{p}{\beta}).
                \end{align*}
 \newpage
                Again, we have two exhaustive cases, depending on $\mi{color}(A.\mi{cmd})$:
                \begin{enumerate}

                    \item \textbf{Case $\mi{color}(A.\mi{cmd}) \neq \mi{white}$.}
                        Note that it must be $\alpha = \alpha'$ since they coincide both on commands (by definition of $\confrel{}{p}{}$) and on the store.
                        Also, by the premises of Rule~\rulename{(Colored)} we have $\cmdrelc{-, \wvar{pc}}{\wwhile b \wdo c''}{\mtt{ls}[n_0], m'}$ and $\cmdrelw{n_0, \wvar{pc}}{c_1 \wseq c_2}{\mtt{ls}[n], m}$. 
                        Since $\ns{(c_1 \wseq c_2, \sigma)}{\alpha} = \ns{(c_1, \sigma)}{\alpha'}$, the operational semantics is deterministic and $\alpha = \alpha'$, we have that  $\beta = \beta'$.
                        So, since $\mi{color}(c'_1) \neq \mi{white}$, to prove that $\confrel{c'_1 \wseq c_2, \sigma}{p}{\beta}$, it remains to prove the following:
                        \begin{itemize}
                            \item $\cmdrelc{-, \wvar{pc}}{\wwhile b \wdo c''}{\mtt{ls}[n_0], m''}$ holds by hypothesis with $m'' = m'$;
                            \item $\cmdrelw{n_0, \wvar{pc}}{c'_1 \wseq c_2}{\mtt{ls}[n'], m}$ that  follows by (IHP) that guarantees that $\cmdrelw{n_0, \wvar{pc}}{c'_1}{\mtt{ls}[n_1], m_1}$ and by the condition $ \confrel{c_1 \wseq c_2, \sigma}{p}{\alpha}$ that ensures $\cmdrelw{n_0, \wvar{pc}}{c_2}{\mtt{ls}[m_1], m}$.
                        \end{itemize}

                        \item \textbf{Case $\mi{color}(A.\mi{cmd}) = \mi{white}$.}
                        Analogous to the case above.
                \end{enumerate}
        \end{description}

        \item \underline{$(\forall B, \alpha.\ \step{A}{}{B} \land \confrel{A}{p}{\alpha} \land \ns{A}{\alpha} = 0 \implies \sz{B} < \sz{A})$.}
        By construction of the measure function $\sz{\cdot}$.
        \item \underline{For any final source configuration $\src{A}$ and obfuscated configuration $\alpha$ there exists a final} \\
        \underline{ obfuscated configuration $\beta$ such that $\stepn{\alpha}{}{\ns{A}{\alpha}}{\beta} \implies \confrel{A}{p}{\beta}$.}
        Trivial.
    \end{itemize}
\end{proof}

We then show that the control-flow flattening obfuscation preserves the constant-time programming policy.
Following~\cite{barthe2018secure}, we show that $(\eqc, \eqc)$ is a \emph{final CT-diagram} w.r.t.\ $\confrel{}{p}{}$, $\ns{\cdot}{\cdot}$ and $\sz{\cdot}$.

To prove that $\eqc$ adheres to the definitions above we need two lemmata:
\begin{lemma}\label{lemma:eqc-t-t-implies-eqc}
    Let $A, A', B, B'$ be source or target configurations.
    If $A \eqc A'$, $\step{A}{t}{B}$ and $\step{A'}{t}{B'}$ then $B \eqc B'$.
\end{lemma}
\begin{proof}
    Follows directly by case analysis on $A.\mi{cmd} = A'.\mi{cmd}$ (equality follows from Definition~\ref{def:eqc}).
\end{proof}
\begin{lemma}\label{lemma:eqc-implies-eqns}
    If $A \eqc A'$, $\alpha \eqc \alpha'$ then $\ns{A}{\alpha} = \ns{A'}{\alpha'}$.
\end{lemma}
\begin{proof}
    This lemma follows from the fact that the function $\ns{\cdot}{\cdot}$ is defined explicitly in the proof of Theorem~\ref{thm:gensim} and depends just on the syntax of $\src{A}$, $A'$, $\alpha$, and $\alpha'$.
\end{proof}

Finally, we can show the following theorem:
\thmctsim*
\begin{proof} (Sketch)
    \begin{enumerate}
        \item \underline{$(\eqc, \eqc)$ is a \emph{manysteps CT-diagram}.}
        The definition of \emph{manysteps CT-diagrams} configurations that, if
        \begin{itemize}
            \item $A \eqc_s A'$ and $\alpha \eqc_t \alpha'$;
            \item $\step{A}{t}{B}$ and $\step{A'}{t}{B'}$;
            \item $\stepn{\alpha}{\tau}{\ns{A}{\alpha}}{\beta}$ and $\stepn{\alpha'}{\tau'}{\ns{A'}{\alpha'}}{\beta'}$;
            \item $\confrel{A}{p}{\alpha}$, $\confrel{A'}{p}{\alpha'}$, $\confrel{B}{p}{\beta}$ and  $\confrel{B'}{p}{\beta'}$
        \end{itemize}
        then
        \begin{itemize}
            \item $\tau = \tau'$ and $\ns{A}{\alpha} = \ns{A'}{\alpha'}$;
            \item $B \eqc B'$ and $\beta \eqc \beta'$;
        \end{itemize}
        The equality of $\tau$ and $\tau'$ directly follows from the fact that $\alpha$ and $\alpha'$ are syntactically the same by hypothesis and are the obfuscated version of two configurations that generate the same observable $t$.
        From Lemma~\ref{lemma:eqc-implies-eqns} we can derive $\ns{A}{\alpha} = \ns{A'}{\alpha'}$.
        Finally, Lemma~\ref{lemma:eqc-t-t-implies-eqc} entails the last two theses.

        \item \underline{$\forall p, \sigma, \sigma'.\ \phi((p, \sigma), (p, \sigma')) \Rightarrow (p, \sigma) \eqc (p, \sigma') \land (\comp{p}, \sigma) \eqc (\comp{p}, \sigma)$.}
        Follows from the definition of $\eqc$ that just requires syntactic equality between configurations.

        \item \underline{$A \eqc_s A' \Rightarrow A \in S_f \iff A' \in S_f$}.
        Again, follows directly from definition of $\eqc$ and of $S_f$.

        \item \underline{$(\eqc, \eqc)$ is a \emph{final CT-diagram}.}
        The definition of \emph{final CT-diagrams} configurations that, if
        \begin{itemize}
            \item $A \eqc A'$ and $\alpha \eqc \alpha'$,
            \item $\src{A}$ and $A'$ are final,
            \item $\stepn{\alpha}{\tau}{\ns{A}{\alpha}}{\beta}$ and $\stepn{\alpha'}{\tau'}{\ns{A'}{\alpha'}}{\beta'}$,
            \item $\confrel{A}{p}{\alpha}$, $\confrel{A'}{p}{\alpha'}$, $\confrel{B}{p}{\beta}$ and $\confrel{B'}{p}{\beta'}$
        \end{itemize}
        then
        \begin{itemize}
            \item $\tau = \tau'$ and $\ns{A}{\alpha} = \ns{A'}{\alpha'}$;
            \item $\beta \eqc \beta'$ and they are both final.
        \end{itemize}
        Since $\src{A}$ and $A'$ are final and $\confrel{A}{p}{\alpha}$ and $\confrel{A'}{p}{\alpha'}$, it must be that $\wvar{pc}$ is $0$ in both $\alpha$ and $\alpha'$.
        Thus, $\alpha$ and $\alpha'$ terminate with $\tau = \tau'$ that just include the check of the \wwhile condition.
        The other theses can be derived following the same proof structure as above.
    \end{enumerate}
\end{proof}